\documentclass[journal]{IEEEtran}

\usepackage{times}
\usepackage{hyperref}
\usepackage{url}
\usepackage{mathrsfs}
\usepackage{bbm}
\usepackage{amsthm}
\usepackage{outlines, mathtools, outlines, cancel}
\usepackage{subcaption}
\usepackage[ruled,vlined]{algorithm2e}
\usepackage[english]{babel}
\usepackage[utf8]{inputenc}
\usepackage[T1]{fontenc}
\usepackage{newtxtext,newtxmath}
\usepackage{array,booktabs}

\newcommand{\norm}[1]{\lVert#1\rVert}

\newcommand{\zer}[1]{\text{Zer}(#1)}
\newcommand{\fix}[1]{\text{Fix}(#1)}
\newcommand{\blkd}[1]{\text{blkd}(#1)}
\DeclarePairedDelimiter\floor{\lfloor}{\rfloor}
\newcommand{\supsub}[3]{{#1}^{\scriptscriptstyle #2}_{\scriptscriptstyle #3}}

\DeclareMathOperator{\minimize}{minimize}
\DeclareMathOperator{\subj}{subject\:to}
\DeclareMathOperator{\optT}{\mathcal{T}}
\DeclareMathOperator{\optA}{\mathcal{A}}
\DeclareMathOperator{\optB}{\mathcal{B}}

\newcommand{\agentN}{\mathcal{N}}
\newcommand{\edgeE}{\mathcal{E}}
\newcommand{\lagrg}{\mathcal{L}}
\newcommand{\feaset}{\mathcal{F}}
\newcommand{\glbset}{\mathcal{Q}}
\newcommand{\feasetx}{\mathcal{X}}
\newcommand{\commg}{\mathcal{G}}
\newcommand{\by}{\boldsymbol y}

\newcommand{\bd}{\boldsymbol d}

\newcommand{\bone}{\boldsymbol 1}
\newcommand{\bzero}{\boldsymbol 0}
\newcommand{\btau}{\boldsymbol \tau}

\newcommand{\rset}[2]{\mathbb{R}^{#1}_{#2}}
\newcommand{\pdset}[2]{\mathbb{S}^{#1}_{#2}}
\newcommand{\crset}[2]{\overline{\mathbb{R}}^{#1}_{#2}}
\newcommand{\nset}[2]{\mathbb{N}^{#1}_{#2}}

\newcommand{\neighbN}[2]{\mathcal{N}^{#1}_{#2}}

\newcommand{\axlW}{W}
\newcommand{\incdtBl}{B_l}
\newcommand{\lagrgnet}{\mathcal{L}_{\text{net}}}
\DeclareMathOperator{\graph}{gra}
\newtheorem{theorem}{Theorem}

\newtheorem{assumption}{Assumption}
\newtheorem{lemma}{Lemma}

\title{\LARGE \bf
A Primal Decomposition Approach to Globally Coupled Aggregative Optimization over Networks
}

\author{Yuanhanqing Huang and Jianghai Hu
\thanks{This work was supported by the National Science Foundation under Grant No. 2014816.}
\thanks{The authors are with the Elmore Family School of Electrical and Computer Engineering, Purdue University, West Lafayette, IN, 47907, USA
        {\tt\small \{huan1282, jianghai\}@purdue.edu}}%
}

\makeatletter
\def\ps@IEEEtitlepagestyle{%
  \def\@oddfoot{\mycopyrightnotice}%
  \def\@evenfoot{}%
}
\def\mycopyrightnotice{
  {\footnotesize
  \begin{minipage}{\textwidth}
  \centering
© 2021 IEEE. Personal use of this material is permitted. Permission from IEEE must be obtained for all other uses, in any current or future media, including reprinting/republishing 
this material for advertising or promotional purposes, creating new collective works, for resale or redistribution to servers or lists, or reuse of any copyrighted component of this work in other works.
  \end{minipage}
  }
}
\let\old@ps@IEEEtitlepagestyle\ps@IEEEtitlepagestyle
\def\confheader#1{%
    \def\ps@IEEEtitlepagestyle{%
        \old@ps@IEEEtitlepagestyle%
        \def\@oddhead{\strut\hfill#1\hfill\strut}%
        \def\@evenhead{\strut\hfill#1\hfill\strut}%
    }%
    \ps@headings%
}
\confheader{%
\footnotesize
\begin{minipage}{\textwidth}
\centering
This article is published as Huang and J. Hu, “A Primal Decomposition Approach to Globally Coupled Aggregative Optimization over Networks,” in 2021 IEEE 60th Conference on Decision and Control (CDC). IEEE, 2021, to appear.
\end{minipage}
}

\begin{document}

\maketitle


\begin{abstract}

We consider a class of multi-agent optimization problems, where each agent has a local objective function that depends on its own decision variables and the aggregate of others, and is willing to cooperate with other agents to minimize the sum of the local objectives.  
After associating each agent with an auxiliary variable and the related local estimates, we conduct primal decomposition to the globally coupled problem and reformulate it so that it can be solved distributedly. 
Based on the Douglas-Rachford method, an algorithm is proposed which ensures the exact convergence to a solution of the original problem. 
The proposed method enjoys desirable scalability by only requiring each agent to keep local estimates whose number grows linearly with the number of its neighbors. 
We illustrate our proposed algorithm by numerical simulations on a commodity distribution problem over a transport network. 
\end{abstract}

\section{Introduction}

In a cooperative multi-agent system, there exists a group of agents each of whom has a specific objective function depending on the joint decision profile of all agents, and they cooperate with each other to optimize the sum of their local objectives. 
Over the past decade, considerable attention and effort have been paid to the consensus optimization problem \cite{yang2019survey, nedic2018distributed}. 
There is also some existing work where each agent keeps its own distinct decision variables \cite{tang2020zeroth, hu2018distributed}. 
In this paper, we restrict our attention to a special case where the influences of other agents' strategies can be represented through some aggregative coupling structures \cite{jensen2010aggregative}. 
The aggregative coupling structures have been used to model numerous applications, e.g., network congestion control \cite{barrera2014dynamic}, demand side management in smart grids \cite{wei2021mechanism}, and charging control of electric vehicles \cite{ma2011decentralized}. 
Besides the aggregative coupling in the objective functions, in many circumstances, the decisions of the agents may be subject to some global resource constraints \cite{camisa2021distributed, falsone2017dual}, such as total energy and communication channel capacity \cite{heydaribeni2019distributed}. 
The coupled objectives and strategy sets of the agents are at odds with local privacy concerns and limited scalability. 
Thus, distributed algorithms are preferred to solve such problems which only allow local exchanges of information. 

Our proposed solution to the above problem is inspired by some recent work in non-cooperative games on networks, i.e., the generalized Nash equilibrium problem (GNEP) \cite{facchinei2010generalized}, which has attracted increasing research interest, especially through the avenue of operator splitting \cite{belgioioso2018douglas, yi2019operator, pavel2019distributed}. 
For example, the algorithms proposed in \cite{parise2019distributed,parise2020distributed} carry out multiple rounds of communication within each iteration. 
With sufficient rounds of information exchange, the proposed algorithms can converge to an $\epsilon$-neighborhood of a generalized Nash equilibrium (GNE). 
The authors of \cite{liang2017distributed} design a continuous-time algorithm based on the projected dynamics and non-smooth tracking strategy, which is only applicable if the coupling constraints can be expressed as a system of linear equations. 
More recently, \cite{gadjov2020single, bianchi2020fast} introduce local estimates of the aggregates of interest, and then leverage the forward-backward splitting and proximal-point algorithms to compute the GNEs, respectively. 
The authors of \cite{belgioioso2020distributed} further develop an algorithm that can deal with time-varying communication networks by integrating the projected pseudo-gradient scheme with dynamic tracking. 
The convergence of these methods relies on a proper initialization and an invariance property throughout the algorithm iterations to ensure valid estimates, which make them vulnerable to system noise and malicious attacks. 

In this paper, we consider cooperative multi-agent optimization problems that are globally coupled by some aggregates as well as some affine global constraints.
We focus on the cases where the aggregates in objectives and global constraints share the same linear functional form. 
Our main contributions are as follows:
a) we present a primal decomposition scheme that converts the original globally coupled problem into local problems among individual agents, subject to some consensus constraints. 
We show that we can find a minimizer of the original problem by computing a zero of an operator derived from the decomposed problem; 
b) we use the Douglas-Rachford (DR) splitting method to develop a distributed algorithm for computing a zero of the previously derived operator. 
The exact convergence of the algorithm can be established without the need for the invariance property.

\textit{Basic Notations:} 
For a set of matrices $\{V_i\}_{i \in S}$, we let $\blkd{V_1, \ldots, V_{|S|}}$ or $\blkd{V_i}_{i \in S}$ denote the diagonal concatenation of these matrices, $[V_1, \ldots, V_{|S|}]$ their horizontal stack, and $[V_1; \cdots; V_{|S|}]$ their vertical stack. 
For a set of vectors $\{v_i\}_{i \in S}$, $[v_i]_{i \in S}$ or $[v_1; \cdots; v_{|S|}]$ denotes their vertical stack. 
For a vector $v$ and a positive integer $i$, $[v]_i$ denotes the $i$th entry of $v$. 
Denote $\crset{}{} \coloneqq \rset{}{} \cup \{+\infty\}$, $\rset{}{+} \coloneqq [0, +\infty)$, and $\rset{}{++} \coloneqq (0, +\infty)$. 
$\pdset{n}{+}$ (resp. $S^n_{++}$) represents the set of all $n\times n$ symmetric positive semi-definite (resp. definite) matrices.
$\iota_{\mathcal{S}}(x)$ is defined to be the indicator function of a set $\mathcal{S}$, i.e., if $x \in \mathcal{S}$, then $\iota_{\mathcal{S}}(x) = 0$; otherwise, $\iota_{\mathcal{S}}(x) = +\infty$. 
$N_{S}(x)$ denotes the normal cone to the set $S \subseteq \rset{n}{}$ at the point $x$: if $x \in S$, then $N_S(x) \coloneqq \{u \in \rset{n}{} \mid \sup_{z \in S} \langle u, z-x \rangle \leq 0 \}$; otherwise, $N_S(x) \coloneqq \varnothing$. 
We use $\rightrightarrows$ to indicate a point-to-set map. 
For an operator $T: \rset{n}{} \rightrightarrows \rset{n}{}$, $\zer{T} \coloneqq \{x \in \rset{n}{} \mid Tx \ni \bzero\}$ and $\fix{T} \coloneqq \{x \in \rset{n}{} \mid Tx \ni x\}$ denote its zero set and fixed point set, respectively. 

\section{Problem Formulation And Preliminaries}

\subsection{Cooperative Multi-Agent Optimization Problem}
We consider a group of agents indexed by $\agentN = \{1, \ldots, N\}$, where each agent $i \in \agentN$ shall choose its decision variables $x_i$ from its local feasible set $\mathcal{X}_i \subseteq \rset{n_i}{}$. 
The goal of each agent $i$ is to minimize its objective $J_i(x_i, s_i(x_{-i}))$, which depends both on its own local decision $x_i$ and the aggregate of other agents' decisions $s_i(x_{-i})$. 
We let the vector $x_{-i}$ represent the vertical stack of other agents' decisions. 
The aggregate $s_i:\rset{n_{-i}}{} \to \rset{l}{}$ is assumed to be of the form $s_i(x_{-i}) \coloneqq \sum_{j\in \neighbN{}{-i}} A_jx_j$, where $A_j \in \rset{l \times n_j}{}$, $n_{-i} \coloneqq \sum_{j \in \neighbN{}{-i}}n_j$, and $\neighbN{}{-i}$ denotes the set of all agents except $i$. 
Besides the local feasible sets, the decisions of all agents should also satisfy the global constraints given by $\tilde{\feasetx} \coloneqq \{x \in \mathcal{X} \mid \sum_{i \in \agentN}A_i x_i \leq c\}$,
where $x \coloneqq [x_1; \cdots; x_N]$, $\mathcal{X} \coloneqq \prod_{i \in \agentN} \mathcal{X}_i$, and $c \in \rset{l}{}$ is a constant vector denoting the total availability of $l$ global resources shared among all involved agents. 

The setting above gives rise to the specific formulations of the multi-agent optimization problem we are going to study. 
In this problem, this group of agents aim to cooperatively solve the following convex optimization problem: 
\begin{align}\label{eq:probsetup}
\begin{cases}
\underset{x_i \in \mathcal{X}_i, i \in \agentN}{\minimize} &\; J(x) \coloneqq \sum_{i \in \agentN}J_i(x_i, s_i(x_{-i})) \\
\subj & \; \sum_{i \in \agentN} A_ix_i \leq c.
\end{cases}
\end{align}

\begin{assumption}{(Existence of Subgradient)}\label{asp:subgrad} 
For each agent $i \in \agentN$, $J_i(x_{i}, s_i):\mathcal{X}_i \times \rset{l}{}\to\crset{}{}$ is an extended-real-valued closed convex proper (CCP) function in $x_i$ and $s_i$.
\end{assumption}
\begin{assumption}{(Feasible Sets)}\label{asp:fesbset}
Each local feasible set $\mathcal{X}_i$ is nonempty, closed and convex. 
The set $\tilde{\mathcal{X}}$ is nonempty and satisfies Slater's constraint qualification \cite[Sec.~3.2]{facchinei2007finite}. 
\end{assumption}

Under Assumptions~\ref{asp:subgrad} and \ref{asp:fesbset}, the problem \eqref{eq:probsetup} admits a nonempty minimizer set if and only if there exists a multiplier $\lambda \in \rset{l}{}$ such that the Karush-Kuhn-Tucker (KKT) system below holds \cite[Sec.~1.3]{facchinei2007finite}: 
\begin{equation}\label{eq:kkt-mini}
\begin{split}
& \partial_{x_i} J(x) + A_i^T\lambda + N_{\mathcal{X}_i}(x_i) \ni \bzero, \; \forall i \in \agentN, \\
& \bzero \leq \lambda \perp c - \textstyle{\sum}_{i \in \agentN}A_ix_i \geq \bzero,
\end{split}
\end{equation}
where the subgradient can be more explicitly written as $\partial_{x_i}J(x) = \partial_{x_i}J_i(x_i, s_i(x_{-i})) + \sum_{j \in \neighbN{}{-i}}A_i^T\partial_{s_j}J_j(x_j, s_j(x_{-j}))$.

\begin{assumption}\label{asp:e&uNE}
The set of minimizers of \eqref{eq:probsetup}, i.e., $\{x \mid J(x) < +\infty, J(x) \leq J(x'), \forall x' \in \tilde{\mathcal{X}}\}$, is nonempty. 
\end{assumption}

A (typically sparse) communication graph $\commg = (\agentN, \edgeE)$ is assumed to exist to implement local information exchanges, where $\edgeE \subseteq \agentN \times \agentN$ denotes the set of directed edges. 
Let $E$ be the cardinality of $\edgeE$. We denote by $(i, j)$ a directed edge with agent $i$ as its tail and agent $j$ as its head. 
Each agent can communicate with its neighbors through arbitrators on the incident edges and then update its local decision variables accordingly. 
For brevity of notation, define the sets of agent $i$'s in- and out-neighbors as $\neighbN{+}{i}\coloneqq \{j\in\mathcal{N}\mid(j, i)\in\mathcal{E}\}$ and $\neighbN{-}{i}\coloneqq \{j\in\mathcal{N}\mid(i, j)\in\mathcal{E} \}$, the cardinalities of which are denoted by $\supsub{N}{+}{i}$ and $\supsub{N}{-}{i}$,  respectively. 
Note that although the multipliers we are going to introduce are defined in a directed fashion, we assume each node can send messages to both its in- and out-neighbors, and $\mathcal{G}$ should satisfy the following assumption. 

\begin{assumption}{(Communicability)}\label{asp:comm}
The underlying communication graph $\mathcal{G} = (\mathcal{N}, \edgeE)$ is undirected and connected. Besides, it has no self-loops and $\forall i \in \agentN, \; \neighbN{-}{i} \neq \varnothing$. 
\end{assumption}

\subsection{The Douglas-Rachford (DR) Splitting Method}

A set-valued operator $\optT: \rset{n}{} \to 2^{\rset{n}{}}$ is called maximally monotone if $\forall (x, u) \in \graph(\optT)$ and $\forall (y, v) \in \graph(\optT)$, $\langle x- y, u-v\rangle \geq 0$, and its graph is not properly contained in the graph of any other monotone operators. 
By leveraging the resolvent of $\optT$, i.e., $J_{\optT} \coloneqq (I + \optT)^{-1}$, the proximal point iteration can generate a sequence converging to a zero of $\optT$ \cite[Sec.~23.4]{BauschkeHeinzH2017CAaM}.
However, this method has a drawback that the resolvent of $\optT$ often can not be easily evaluated. Another potential difficulty is that in a network with multiple agents, $J_{\optT}$ usually cannot be implemented distributedly. 

The DR splitting method decomposes the operator $\optT$ as the sum of two other operators $\optA$ and $\optB$, whose resolvents are easier to evaluate. This can thus alleviate the computational burden by evaluating $J_{\optA} \coloneqq (I + \optA)^{-1}$ and $J_{\optB} \coloneqq (I + \optB)^{-1}$ instead.  Moreover, if $\zer{\optT} \neq \varnothing$, the convergence of the DR method can be obtained under very mild assumptions, i.e., $\optA$ and $\optB$ are set-valued and maximally monotone \cite[Sec.~26.3]{BauschkeHeinzH2017CAaM}. 

A common strategy to solve networked problems is that, rather than focusing on $\optT(\psi) \ni \bzero$, we consider $\Phi^{-1} \optT(\psi) \ni \bzero$. 
Here, $\Phi$ is a positive definite matrix, called the design matrix, which is introduced to facilitate the distributed implementation. 
By applying the DR method to the splitting $\Phi^{-1} \optA + \Phi^{-1} \optB$, we obtain the following updating steps:
\begin{equation}\label{eq:D-R}
\begin{split}
& \text{Calculate}\; J_{\Phi^{-1}\optA}: \supsub{\psi}{(k+1)}{} \coloneqq J_{\Phi^{-1}\optA}(\supsub{\tilde{\psi}}{(k)}{});\\
& \text{R-R updates}: \supsub{\hat{\psi}}{(k+1)}{} \coloneqq 2 \cdot \supsub{\psi}{(k+1)}{} - \supsub{\tilde{\psi}}{(k)}{}; \\
& \text{Calculate}\; J_{\Phi^{-1}\optB}: \supsub{\bar{\psi}}{(k+1)}{} \coloneqq J_{\Phi^{-1}\optB}(\supsub{\hat{\psi}}{(k+1)}{}); \\
& \text{K-M updates}: \supsub{\tilde{\psi}}{(k+1)}{} \coloneqq \supsub{\tilde{\psi}}{(k)}{} + 2\supsub{\gamma}{(k)}{}(\supsub{\bar{\psi}}{(k+1)}{} - \supsub{\psi}{(k+1)}{}).
\end{split}
\end{equation}
In addition to the aforementioned conditions, if $(\supsub{\gamma}{(k)}{})_{k \in \nset{}{}}$ is a sequence in $[0, 1]$ satisfying $\sum_{k \in \nset{}{}}\supsub{\gamma}{(k)}{}(1 - \supsub{\gamma}{(k)}{}) = +\infty$, then the sequence $(\supsub{\psi}{(k)}{})_{k\in\nset{}{}}$ generated by \eqref{eq:D-R} will converge to a zero of $\optT$ \cite[Thm.~26.11]{BauschkeHeinzH2017CAaM}.  
The second and fourth steps in \eqref{eq:D-R} are trivial and the major workload resides in the first and third steps. 
We use $\supsub{\hat{\psi}}{(k)}{}$ to denote the results of the reflected resolvent (R-R) updates  $R_{\Phi^{-1}\optA}(\supsub{\tilde{\psi}}{(k)}{}) = 2\supsub{\psi}{(k)}{} - \supsub{\tilde{\psi}}{(k)}{}$ and $\supsub{\tilde{\psi}}{(k+1)}{}$ the results of the Krasnoselskij-Mann (K-M) updates $\supsub{\tilde{\psi}}{(k)}{} + 2\supsub{\gamma}{(k)}{}(\supsub{\bar{\psi}}{(k+1)}{} - \supsub{\psi}{(k+1)}{})$.
This set of notations is used throughout for brevity, and similar notations are defined for decisions $\supsub{x}{(k)}{}$ $(\supsub{\hat{x}}{(k)}{}, \supsub{\bar{x}}{(k)}{}, \supsub{\tilde{x}}{(k)}{})$, multipliers $\supsub{\lambda}{(k)}{}$ $(\supsub{\hat{\lambda}}{(k)}{}, \supsub{\bar{\lambda}}{(k)}{}, \supsub{\tilde{\lambda}}{(k)}{})$, etc.

\section{Distributed Algorithm Using Primal Decomposition}

\subsection{Primal Decomposition of the Problem}

Given the multi-agent optimization problem~\eqref{eq:probsetup}, we introduce for each agent $i$ a variable $\sigma_i \coloneqq \sum_{j \in \neighbN{}{-i}}A_jx_j \in \rset{l}{}$ to track the aggregate of other agents. 
The global problem~\eqref{eq:probsetup} can be equivalently recast into the following problems:
\begin{equation}\label{eq:probsetup3-1}
\small
\begin{split}
& (\forall i \in \agentN) \;
\begin{cases}
\underset{x_i \in \mathcal{X}_i, \sigma_i}{\minimize} & J_i(x_i, \sigma_i) \\
\subj & A_ix_i + \sigma_i \leq c;
\end{cases} \\ 
& \text{all agents collectively satisfy: } \sigma_i = {\textstyle\sum}_{j \in \neighbN{}{-i}}A_jx_j, \forall i \in \agentN. 
\normalsize
\end{split}
\end{equation}

The equivalent relationship between the problem \eqref{eq:probsetup} and \eqref{eq:probsetup3-1} is given in the following lemma for sake of clarity. 
\begin{lemma}\label{le:equiv-initial}
Suppose Assumptions~\ref{asp:subgrad} and \ref{asp:fesbset} hold. Then for any stack vector $x^* \coloneqq [x^*_i]_{i \in \agentN}$ where each $\{x^*_i, \sum_{j \in \neighbN{}{-i}}A_jx^*_j\}$ is a minimizer of agent $i$'s recast problem~\eqref{eq:probsetup3-1}, there exists a proper $\lambda^* \in \rset{l}{}$ such that $(x^*, \lambda^*)$ is a solution of \eqref{eq:kkt-mini}. 
Conversely, if the KKT system \eqref{eq:kkt-mini} admits a solution $(x^\dagger, \lambda^\dagger)$, where $x^\dagger \coloneqq [x^\dagger_i]_{i \in \agentN}$, for each agent $i \in \agentN$, $\{x^\dagger_i, \sum_{j \in \neighbN{}{-i}}A_jx^\dagger_j\}$ is a minimizer of its optimization problem~\eqref{eq:probsetup3-1}. 
\end{lemma}

\begin{proof}
See Appendix~\ref{pf:equiv-initial}. 
\end{proof}

However, the current formulation of $\sigma_i$ is still globally dependent, while we need to estimate this aggregate with only local communications. 
To this end, we associate a weight matrix $W \in \rset{N \times N}{}$ with the communication graph $\commg$ that satisfies the conditions in Assumption \ref{asp:comm}. 
If there is a directed edge from agent $i$ to agent $j$ ($j \neq i$), $W_{ij}$ is assigned some proper negative value; otherwise, $W_{ij} = 0$. 
The diagonal entries are set to be $W_{ii} = -\sum_{j \in \agentN} W_{ij} = -\sum_{j \in \neighbN{-}{i}}W_{ij} > 0$. 
By the connectivity condition in Assumption \ref{asp:comm}, the defined weight matrix $W$ has an eigenvalue $0$ with multiplicity $1$, with an associated eigenvector $\bone_N$. 
Furthermore, we endow each agent $i$ with an auxiliary variable $y_i \in \rset{l}{}$.
With the introduction of $\{y_i\}_{i \in \agentN}$ and weight matrix $W$, we can rewrite the problem \eqref{eq:probsetup3-1} as follows: 
\begin{equation}\label{eq:probsetup3-2}
\begin{cases}
    \underset{x_i \in \mathcal{X}_i, \sigma_i, y_i}{\minimize} \; J_i(x_i, \sigma_i) \\
    \subj \; A_i x_i + \sigma_i \leq c \\
    \qquad\quad (N-1)A_i x_i - \sigma_i = \axlW_{ii}y_i + \sum_{j \in \neighbN{+}{i}}\axlW_{ji} y_{j} \\
    \qquad\quad A_i x_i + \sigma_i = A_j x_j + \sigma_j, \forall j \in \agentN_{i}^{+}.
 \end{cases}
\end{equation}
A proof of the equivalence between \eqref{eq:probsetup3-1} and \eqref{eq:probsetup3-2} can be found in Appendix~\ref{pf:agggamecnst}.
Besides its own variables, the problem~\eqref{eq:probsetup3-2} of agent $i$ only involves the auxiliary variables of its in-neighbors. Hence, \eqref{eq:probsetup3-2} is a locally dependent problem. 
To facilitate the distributed implementation, each agent $i$ is assumed to keep a local estimate for the auxiliary variable of each of its in-neighbors. 
Let $y_{ji}$ denote the local estimate of $y_j$ kept by agent $i$. With these local estimates, we can rewrite the preceding problem \eqref{eq:probsetup3-2} as: 
\begin{equation}\label{eq:probsetup3-3}
    \begin{cases}
    \underset{x_i \in \mathcal{X}_i, \sigma_i, \by_i}{\minimize} \; J_i(x_i, \sigma_i) \\
    \subj \; A_i x_i + \sigma_i \leq c \\
    \qquad\quad (N-1)A_i x_i - \sigma_i = \axlW_{ii}y_i + \sum_{j \in \neighbN{+}{i}}\axlW_{ji} y_{ji} \\
    \qquad\quad A_i x_i + \sigma_i = A_j x_j + \sigma_j, \forall j \in \agentN_{i}^{+} \\
    \qquad\quad y_{ji} = y_{j}, \forall j \in \neighbN{+}{i}, 
    \end{cases}
\end{equation}
where $\by_i \coloneqq [y_i; [y_{ji}]_{j \in \neighbN{+}{i}}]$. 
By this reformulation, the coupling with other agents is sandboxed inside the consensus constraints between each auxiliary variable $y_i$ and the local estimate $y_{ij}$ kept by each out-neighbor $j$. 
As will be shown later, these consensus constraints can be enforced by some simple updating steps, making \eqref{eq:probsetup3-3} more tractable than the one in \eqref{eq:probsetup3-2}. 
It is worth highlighting that the local variables kept by each agent $i$ has dimension $n_i + l(2 + \supsub{N}{+}{i})$ which merely grows linearly with the number of its in-neighbors and allows better scalability in a sparsely connected communication network. 

To study agent $i$'s local optimization problem \eqref{eq:probsetup3-3}, we consider its associated Lagrangian $\lagrg_{i}$ defined as follows:
\small
\begin{align*}
\lagrg_{i} &= J_i(x_i, \sigma_i) + \iota_{\mathcal{X}_i}(x_i) + \iota_{\glbset_i}(x_i, \sigma_i) + \iota_{\feaset_i}(x_i, \sigma_i, \by_i)\\
& + {\textstyle\sum}_{j \in \neighbN{+}{i}}\big(\lambda_{ji}^T (A_i x_i + \sigma_i - A_j x_j - \sigma_j) + \mu_{ji}^T (y_{ji} - y_j)\big),
\end{align*}
\normalsize
where $\lambda_{ji}$ is the Lagrange multiplier of the constraint $A_ix_i + \sigma_i - A_jx_j - \sigma_j = 0$; 
$\mu_{ji}$ is the Lagrange multiplier incentivizing the consensus between $y_{ji}$ and $y_{j}$; 
$\glbset_i \coloneqq \{(x_i, \sigma_i) \mid A_ix_i + \sigma_i \leq c\}$ denotes the local feasible set corresponding to the first constraint of \eqref{eq:probsetup3-3};
and $\feaset_i \coloneqq \{(x_i, \sigma_i, \by_i) \mid (N-1)A_ix_i - \sigma_i = W_{ii}y_i + \sum_{j \in \neighbN{+}{i}}W_{ji}y_{ji}\}$ denotes the local feasible set corresponding to the second constraint of \eqref{eq:probsetup3-3}. 

Summing the Lagrangians of all agents, we obtain the Lagrangian for the network optimization problem given by:
\begin{equation}\label{eq:lagrg-net}
\small
\begin{split}
& \lagrg_{\text{net}} = \underset{i \in \agentN}{\sum} \big( J_i(x_i, \sigma_i) + \iota_{\mathcal{X}_i}(x_i) + \iota_{\glbset_i}(x_i, \sigma_i)+ \iota_{\feaset_i}(x_i, \sigma_i, \by_i) \big) \\
& +\underset{(j, i) \in \edgeE}{\sum} \big(\lambda_{ji}^T(A_i x_i + \sigma_i - A_j x_j - \sigma_j) + \mu_{ji}^T(y_{ji} - y_j) \big).
\end{split}
\normalsize
\end{equation}
As will be shown later in the paper, by finding a saddle point of the network Lagrangian, we can obtain a solution of the original problem \eqref{eq:probsetup}. 
For convenience, we shall write $\{\by_i\}$ in replacement of the more cumbersome notation $\{\by_i\}_{i \in \agentN}$ and similarly for other variables on nodes and edges (e.g. $\{\mu_{ji}\}$ in replacement of $\{\mu_{ji}\}_{(j,i) \in \edgeE_g}$), unless otherwise specified. 

Noting the structure of the consensus term in the network Lagrangian defined above, i.e., 
$\sum_{(j, i) \in \edgeE}\mu_{ji}^T(y_{ji} - y_j)$, 
we can organize $\{\by_i\}$ and $\{\mu_{ji}\}$ into a single vector $\omega \coloneqq [\omega_1; \cdots; \omega_i; \cdots; \omega_N]$, and construct a matrix $M_y \coloneqq \blkd{M_{y_i}}_{i \in \agentN}\in \rset{(2E + N) \times (2E+N)}{}$ with $\omega_i$ and $M_{y_i}$ defined by:
\begin{align}
\omega_i \coloneqq \begin{bmatrix}
y_i \\ 
\mu_{ij_1} \\
y_{ij_1} \\
\vdots \\
\mu_{ij_{N^-_i}} \\
y_{ij_{N^-_i}}
\end{bmatrix}, 
M_{y_i} \coloneqq \begin{bmatrix}
0 & -1 & 0 & \cdots & -1 & 0 \\
-1 & 0 & 1 & \cdots & 0 & 0 \\
0 & 1 & 0 & \cdots & 0 & 0 \\
\vdots & \vdots & \vdots & \ddots & \vdots & \vdots \\
-1 & 0 & 0 & \cdots & 0 & 1 \\
0 & 0 & 0 & \cdots & 1 & 0 \\
\end{bmatrix},
\end{align}
where $(j_1,\ldots, j_{N^{-}_{i}})$ is an arbitrary ordered list of $\neighbN{-}{i}$. Then, the consensus term can be written as $\frac{1}{2} \omega^T (M_y \otimes I_l) \omega$. 

To solve the multi-agent optimization problem, we need to find the stationary points of the network Lagrangian \eqref{eq:lagrg-net}. By taking the subgradient of $\lagrgnet$ w.r.t. each variable and reversing the sign of the rows corresponding to the dual variables, we can derive the set-valued operator $\optT$ given by:
\begin{equation}
\begin{split}
& \optT: \psi \mapsto 
\partial\big({\textstyle\sum}_{i \in \agentN} J_i(x_i, \sigma_i) + \iota_{\mathcal{X}_i}(x_i) + \iota_{\glbset_i}(x_i, \sigma_i)  \\
&  + \iota_{\feaset_i}(x_i, \sigma_i, \by_i)\big) +  \begin{bmatrix}
\bzero & \bzero & A^T\incdtBl & \bzero \\
\bzero & \bzero & \incdtBl & \bzero \\
-\incdtBl^TA & -\incdtBl^T & \bzero & \bzero \\
\bzero & \bzero & \bzero & M_y'
\end{bmatrix}\psi, 
\end{split}
\end{equation}
where $A \coloneqq \blkd{A_i}_{i \in \agentN}$; 
$B$ is the incidence matrix of the communication graph $\mathcal{G}$; $B_l \coloneqq (B \otimes I_l)$; 
$M_y' \coloneqq \blkd{M_{y_i}'}_{i \in \agentN}$ with $M_{y_i}'$ being a skew-symmetric matrix generated by reversing the sign of the even rows of $M_{y_i}$; 
$\sigma$ and $\lambda$ are the stack vectors of $\{\sigma_i\}$ and $\{\lambda_{ji}\}$; $\psi$ denotes the stack of the these variables, i.e. $\psi \coloneqq [x; \sigma; \lambda; \omega]$.

\begin{theorem}\label{thm:agggamecnst}
Suppose Assumptions \ref{asp:subgrad}, \ref{asp:fesbset}, and \ref{asp:comm} hold. Then for any zero point $[x^*; \sigma^*; \lambda^*; \omega^*] \in \zer{\optT}$, $x^*$ is a minimizer of (\ref{eq:probsetup}). Conversely, for any minimizer $x^\dagger$ of the problem (\ref{eq:probsetup}), we can choose proper $\sigma^\dagger$, $\lambda^\dagger$, and $\omega^\dagger$, such that $[x^\dagger; \sigma^\dagger; \lambda^\dagger; \omega^\dagger] \in \zer{\optT}$.
\end{theorem}
\begin{proof}
See Appendix~\ref{pf:agggamecnst}.
\end{proof}

\subsection{Operator Splitting and Distributed Algorithm}

We can split the operator $\optT$ into two operators $\optA$ and $\optB$, and construct a design matrix $\Phi$, which are given by:
\begin{equation}
\begin{split}
\optA : \begin{bmatrix} x \\ \sigma \\ \lambda \\ \omega \end{bmatrix} &\mapsto \partial (\sum_{i \in \agentN} J_i(x_i, \sigma_i)) \\
& + \begin{bmatrix}
\bzero & \bzero & \frac{1}{2}A^TB_l & \bzero \\
\bzero & \bzero & \frac{1}{2}B_l & \bzero \\
-\frac{1}{2}B_l^TA & -\frac{1}{2}B_l^T & \bzero & \bzero \\
\bzero & \bzero & \bzero & M_y'
\end{bmatrix}
\begin{bmatrix} x \\ \sigma \\ \lambda \\ \omega \end{bmatrix},
\end{split}
\end{equation}
\begin{equation}
\begin{split}
\optB : \begin{bmatrix} x \\ \sigma \\ \lambda \\ \omega \end{bmatrix} &\mapsto  \partial \big(\sum_{i \in \agentN}\iota_{\mathcal{X}_i}(x_i) + \iota_{\glbset_i}(x_i, \sigma_i) + \iota_{\feaset_i}(x_i, \sigma_i, \by_i)\big) \\
& + \begin{bmatrix}
\bzero & \bzero & \frac{1}{2}A^TB_l & \bzero \\
\bzero & \bzero & \frac{1}{2}B_l & \bzero \\
-\frac{1}{2}B_l^TA & -\frac{1}{2}B_l^T & \bzero & \bzero \\
\bzero & \bzero & \bzero & \bzero
\end{bmatrix}
\begin{bmatrix} x \\ \sigma \\ \lambda \\ \omega \end{bmatrix},
\end{split}
\end{equation}
\begin{equation}
\Phi \coloneqq 
\begin{bmatrix}
\btau_1^{-1} & \bzero & -\frac{1}{2} A^TB_l & \bzero \\
\bzero & \btau_2^{-1} & -\frac{1}{2} B_l & \bzero \\
-\frac{1}{2} B_l^T A & -\frac{1}{2} B_l^T & \btau_3^{-1} & \bzero \\
\bzero & \bzero & \bzero & \btau_4^{-1}
\end{bmatrix},
\end{equation}
where $\btau_1 \coloneqq \blkd{\tau_{11} I_n, \ldots, \tau_{1N} I_n}$ with the scalars $\supsub{\tau}{}{1i} > 0$ for $i \in \agentN$;
similarly for $\btau_2$, $\btau_3$, and $\btau_4$. 
These step sizes can be chosen based on the Gershgorin circle theorem \cite{bell1965gershgorin} to guarantee $\Phi \in \pdset{}{++}$. 
The operator $\optA$ is maximally monotone since in Assumption~\ref{asp:subgrad}, each objective $J_i$ is assumed to be jointly convex in $x_i$ and $\sigma_i$, and $D + \bar{M}_y'$ is skew-symmetric. 
The maximal monotonicity of $\optB$ can be similarly established. 

As suggested in the DR splitting \eqref{eq:D-R}, we next evaluate the analytical expressions for $J_{\Phi^{-1}\optA}$ and $J_{\Phi^{-1}\optB}$. 
For brevity, let $\lambda^{(k)}_{iB} \coloneqq \sum_{j \in \neighbN{+}{i}} \lambda^{(k)}_{ji} - \sum_{j \in \neighbN{-}{i}}\lambda^{(k)}_{ij}$, which can be obtained through the communications among agent $i$ and its incident edges. The variable $\hat{\lambda}^{(k)}_{iB}$ is defined similarly. 

For $\psi^{(k+1)} \coloneqq J_{\Phi^{-1}\optA}(\Tilde{\psi}^{(k)})$, it corresponds to the inclusion $(\Phi + \optA)\psi^{(k+1)} \ni \Phi \Tilde{\psi}^{(k)}$. 
Each agent $i$ can update its local decisions $x_i$ and local estimates $\sigma_i$ by solving the following problem using local information $\tilde{x}^{(k)}_i$, $\tilde{\sigma}^{(k)}_i$ and the dual information $\tilde{\lambda}^{(k)}_{iB}$ from its incident edges: 
\begin{equation}\label{eq:optA-updt-x-sigma}
\begin{split}
\underset{x_i, \sigma_i}{\minimize} \;& J_i(x_i, \sigma_i) + \tfrac{1}{2}(\tilde{\lambda}^{(k)}_{iB})^T(A_ix_i + \sigma_i) \\
& + \tfrac{1}{2\tau_{1i}}\norm{x_i - \tilde{x}^{(k)}_i}^2_2 + \tfrac{1}{2\tau_{2i}}\norm{\sigma_i - \tilde{\sigma}^{(k)}_i}^2_2.
\end{split}
\end{equation}
After both incident agents solve \eqref{eq:optA-updt-x-sigma}, the dual variable $\lambda^{(k+1)}_{ji}$ maintained by the edge $(j, i)$ can be updated by:
\begin{equation}\label{eq:optA-updt-lambda}
\supsub{\lambda}{(k+1)}{ji} = \supsub{\tilde{\lambda}}{(k)}{ji} + \tfrac{1}{2}\supsub{\tau}{}{3ji} \cdot \big(\supsub{A}{}{i}\supsub{\hat{x}}{(k+1)}{i} + \supsub{\hat{\sigma}}{(k+1)}{i} - \supsub{A}{}{j}\supsub{\hat{x}}{(k+1)}{j} - \supsub{\hat{\sigma}}{(k+1)}{j}),
\end{equation}
where $\supsub{\hat{x}}{(k+1)}{i}$ and $\supsub{\hat{\sigma}}{(k+1)}{i}$ are the results of $R_{\Phi^{-1}\optA}$ as shown in \eqref{eq:D-R}.
The update of $\supsub{\omega}{(k+1)}{}$ is given by $(\btau_4^{-1} + M_y') \cdot \supsub{\omega}{(k+1)}{} = \btau_4^{-1}\supsub{\tilde{\omega}}{(k)}{}$, which is independent of the updates of $x$, $\sigma$ and $\lambda$. 
To have the analytical solution to it, notice that:
\begin{equation}
\begin{split}
& \supsub{\mu}{(k+1)}{ij} + \supsub{\tau}{}{4i}(\supsub{y}{(k+1)}{i} - \supsub{y}{(k+1)}{ij}) = \supsub{\Tilde{\mu}}{(k)}{ij}, \;\;\; \text{edge }(i, j); \\
& \supsub{y}{(k+1)}{i} + \tau_{4i}(-{\textstyle\sum}_{j \in \neighbN{-}{i}} \supsub{\mu}{(k+1)}{ij}) = \supsub{\Tilde{y}}{(k)}{i}, \;\; \text{agent }i; \\
& \supsub{y}{(k+1)}{ij} + \tau_{4i} \supsub{\mu}{(k+1)}{ij} = \supsub{\Tilde{y}}{(k)}{ij}, \qquad \qquad \;\;\; \text{agent }j.
\end{split}
\end{equation}
The above computation is restricted to the agent $i$, its out-edges, and its out-neighbors. We can thus derive the analytical expressions for the local updates of auxiliary variables and their dual variables, which are summarized in Subroutine \ref{algm:axl-csns-updt}.

\begin{algorithm}
\SetAlgorithmName{Subroutine}{subroutine}{List of Subroutines}
\SetAlgoLined
\caption{Auxiliary-Consensus-Update (ACU)}\label{algm:axl-csns-updt}
\textbf{Input:} $\{\supsub{\Tilde{\by}}{(k)}{i}\}$, $\supsub{\{\Tilde{\mu}}{(k)}{ji}\}$;

\For{agent $i \in \agentN$}{
For each $j \in \supsub{\mathcal{N}}{-}{i}$, receive $\supsub{\tilde{\mu}}{(k)}{ij}$ and $\supsub{\tilde{y}}{(k)}{ij}$\;
$\supsub{y}{(k+1)}{i} \coloneqq \frac{1 + \tau_{4i}^2}{1 + \tau_{4i}^2 (1 + N^-_i)}(\supsub{\Tilde{y}}{(k)}{i} + \frac{\tau_{4i}}{1 + \tau_{4i}^2}\underset{j \in \agentN^-_i}{\sum}(\supsub{\Tilde{\mu}}{(k)}{ij} + \supsub{\tau}{}{4i} \supsub{\Tilde y}{(k)}{ij}))$\;
}
\For{edge $(j, i) \in \edgeE$}{
Receive $\supsub{y}{(k+1)}{j}$ from its tail and $\supsub{\tilde{y}}{(k)}{ji}$ from its head \;
$\supsub{\mu}{(k+1)}{ji} \coloneqq \frac{1}{1 + \tau_{4j}^2}\supsub{\Tilde{\mu}}{(k)}{ji} + \frac{\tau_{4j}}{1 + \tau_{4j}^2}(\supsub{\Tilde{y}}{(k)}{ji} - \supsub{y}{(k+1)}{j})$\;
}
\For{agent $i \in \agentN$}{
For each $j \in \supsub{\mathcal{N}}{+}{i}$, receive $\supsub{\mu}{(k+1)}{ji}$ from its in-edge \;
$ \supsub{y}{(k+1)}{ji} \coloneqq \supsub{\Tilde{y}}{(k)}{ji} - \supsub{\tau}{}{4j} \supsub{\mu}{(k+1)}{ji}$\;
}
\textbf{Return:} $\{\supsub{\by}{(k+1)}{i}\}$, $\{\supsub{\mu}{(k+1)}{ji}\}$.
\end{algorithm}

For $\bar{\psi}^{(k+1)} \coloneqq J_{\Phi^{-1}\optB}(\hat{\psi}^{(k+1)})$, the inclusions of $\bar{x}, \bar{\sigma}$, and $\bar{y}$ are coupled locally, and each agent $i$ should update them by solving the following local system of inclusions: 
\begin{equation}\label{eq:xsigmay-incl3}
\begin{split}
\partial_{x_i} \iota_{\glbset_i}(\supsub{\bar{x}}{(k+1)}{i}, \supsub{\bar{\sigma}}{(k+1)}{i}) + 
\partial_{x_i} \iota_{\feaset_i}(\supsub{\bar{x}}{(k+1)}{i}, \supsub{\bar{\sigma}}{(k+1)}{i}, \supsub{\bar{\by}}{(k+1)}{i}) & \\
+ \partial_{x_i} \iota_{\mathcal{X}_i}(\supsub{\bar{x}}{(k+1)}{i}) + \tfrac{1}{2}A_i^T\supsub{\hat{\lambda}}{(k+1)}{i} + \tfrac{1}{\tau_{1i}}(\supsub{\bar{x}}{(k+1)}{i} - \supsub{\hat{x}}{(k+1)}{i})&\ni \bzero \\
\partial_{\sigma_i} \iota_{\glbset_i}(\supsub{\bar{x}}{(k+1)}{i}, \supsub{\bar{\sigma}}{(k+1)}{i}) 
+ \partial_{\sigma_i} \iota_{\feaset_i}(\supsub{\bar{x}}{(k+1)}{i}, \supsub{\bar{\sigma}}{(k+1)}{i}, \supsub{\bar{\by}}{(k+1)}{i}) & \\
+ \tfrac{1}{2}\supsub{\hat{\lambda}}{(k+1)}{i} + \tfrac{1}{\tau_{2i}}(\supsub{\bar{\sigma}}{(k+1)}{i} - \supsub{\hat{\sigma}}{(k+1)}{i})&\ni \bzero \\
\partial_{\by_i} \iota_{\feaset_i}(\supsub{\bar{x}}{(k+1)}{i}, \supsub{\bar{\sigma}}{(k+1)}{i}, \supsub{\bar{\by}}{(k+1)}{i}) + \tfrac{1}{\tau_{4i}}(\supsub{\bar{\by}}{(k+1)}{i} - \supsub{\hat{\by}}{(k+1)}{i})&\ni \bzero.
\end{split}
\end{equation}
Define $M_{\feaset_i} \coloneqq [(N-1)A_i, -I_l, -\axlW_{ii} \otimes I_l, -[\axlW_{ji}]_{\scriptscriptstyle j \in \neighbN{+}{i}}^T \otimes I_l] \in \rset{\scriptscriptstyle l \times (n_i + l(\supsub{N}{+}{i} + 2))}{}$, 
and $\supsub{\check{x}}{(k+1)}{i} \coloneqq \supsub{\hat{x}}{(k+1)}{i} - \frac{\tau_{1i}}{2}A_i^T\supsub{\hat{\lambda}}{(k+1)}{i}$ and $\supsub{\check{\sigma}}{(k+1)}{i} \coloneqq \supsub{\hat{\sigma}}{(k+1)}{i} - \frac{\tau_{2i}}{2}\supsub{\hat{\lambda}}{(k+1)}{i}$. 
Then, finding zeros of \eqref{eq:xsigmay-incl3} is equivalent to solving the following constrained minimization problem:
\begin{equation}\label{eq:optB-updt-xsigy}
\begin{cases}
\underset{x_i\in \mathcal{X}_i, \sigma_i, \by_i}{\minimize} &  \frac{1}{2\tau_{1i}}\norm{x_i - \supsub{\check{x}}{(k+1)}{i}}^2 + \frac{1}{2\tau_{2i}}\norm{\sigma_i - \supsub{\check{\sigma}}{(k+1)}{i}}^2 \\
& \qquad + \frac{1}{2\tau_{4i}}\norm{\by_i - \supsub{\hat{\by}}{(k+1)}{i}}^2 \\
\subj &  M_{\feaset_i}\cdot[x_i ; \sigma_i ; \by_i] = \bzero, A_ix_i + \sigma_i \leq c
\end{cases}.
\end{equation}
For the variables maintained by the edges, the dual variable $\bar{\mu}$ remains the same, while $\bar{\lambda}_{ji}$ is updated by:
\begin{equation}\label{eq:optB-updt-lambda}
\begin{split}
& \supsub{\bar{\lambda}}{(k+1)}{ji} = \supsub{\hat{\lambda}}{(k+1)}{ji} + \supsub{\tau}{}{3ji} \cdot \big(\supsub{A}{}{i}\supsub{\bar{x}}{(k+1)}{i} + \supsub{\bar{\sigma}}{(k+1)}{i} - \supsub{A}{}{j}\supsub{\bar{x}}{(k+1)}{j}  \\
& \quad - \supsub{\bar{\sigma}}{(k+1)}{j} - \tfrac{1}{2}(\supsub{A}{}{i}\supsub{\hat{x}}{(k+1)}{i} + \supsub{\hat{\sigma}}{(k+1)}{i} - \supsub{A}{}{j}\supsub{\hat{x}}{(k+1)}{j} - \supsub{\hat{\sigma}}{(k+1)}{j})\big).
\end{split}
\end{equation}
Overall, Algorithm \ref{algm:multiagt-opt-agg} summarizes the proposed algorithm for finding a minimizer of \eqref{eq:probsetup} based on the DR framework \eqref{eq:D-R}. 

\begin{theorem}
Suppose that Assumptions \ref{asp:subgrad}-\ref{asp:comm} hold, the sequence $(\supsub{\gamma}{(k)}{})_{k \in \nset{}{}}$ satisfies $\supsub{\gamma}{(k)}{} \in [0, 1]$ and $\sum_{k \in \nset{}{}}\supsub{\gamma}{(k)}{}(1 - \supsub{\gamma}{(k)}{}) = +\infty$, and $\btau_1$, $\btau_2$, $\btau_3$ and $\btau_4$ are properly chosen such that the design matrix $\Phi$ is positive definite. Then the sequence $(x^k)_{k \in \nset{}{}}$ generated by Algorithm~\ref{algm:multiagt-opt-agg} will converge to a zero of the operator $\optT$ and thereby to a v-GNE of the original GNEP in \eqref{eq:probsetup}.
\end{theorem}

\begin{proof}
The convergence of the proposed algorithm immediately follows from the fact that the operators $\optA$ and $\optB$ are both maximally monotone \cite[Thm.~26.11]{BauschkeHeinzH2017CAaM}.
\end{proof}

\begin{algorithm}
\small
\SetAlgoLined
\caption{Algorithm for Globally-Coupled Multi-Agent Optimization Problem}\label{algm:agggame}
\textbf{Initialize:} $\{\supsub{\tilde{x}}{(0)}{i}\}, \{\supsub{\tilde{\sigma}}{(0)}{i}\}, \{\supsub{\tilde{\lambda}}{(0)}{ji}\}, \{\supsub{\tilde{\by}}{(0)}{i}\}, \{\supsub{\tilde{\mu}}{(0)}{ji}\}$\;
\textbf{Iterate until convergence:}\\ 
$(\{\supsub{\by}{(k+1)}{i}\}, \{\supsub{\mu}{(k+1)}{ji}\}) \coloneqq \text{ACU}(\{\supsub{\Tilde{\by}}{(k)}{i}\}, \{\supsub{\Tilde{\mu}}{(k)}{ji}\})$ \;
\For{agent $i \in \mathcal{N}$}{
Receive $\supsub{\tilde{\lambda}}{(k)}{ji}$ and $\supsub{\tilde{\lambda}}{(k)}{ij}$ from its in- and out-edges \;
Obtain $\supsub{x}{(k+1)}{i}$ and $\supsub{\sigma}{(k+1)}{i}$ by solving \eqref{eq:optA-updt-x-sigma}\;
R-R updates: $\supsub{\hat{x}}{(k+1)}{i}, \supsub{\hat{\sigma}}{(k+1)}{i}$, $\supsub{\hat{\by}}{(k+1)}{i}$\;
}
\For{edge $(j,i) \in \mathcal{E}$}{
Receive $A_i\supsub{\hat{x}}{(k+1)}{i} + \supsub{\hat{\sigma}}{(k+1)}{i}$ and $A_i\supsub{\hat{x}}{(k+1)}{j} + \supsub{\hat{\sigma}}{(k+1)}{j}$\; 
Obtain $\supsub{\lambda}{(k+1)}{ji}$ by \eqref{eq:optA-updt-lambda}; R-R updates: $\supsub{\hat{\lambda}}{(k+1)}{ji}$\;
}
\For{agent $i \in \mathcal{N}$}{
Receive $\supsub{\hat{\lambda}}{(k+1)}{ji}$ and $\supsub{\hat{\lambda}}{(k+1)}{ij}$ from its in- and out-edges \;
Obtain $\supsub{\bar{x}}{(k+1)}{i}$, $\supsub{\bar{\sigma}}{(k+1)}{i}$ and $\supsub{\bar{\by}}{(k+1)}{i}$ by solving \eqref{eq:optB-updt-xsigy}\;
}
\For{edge $(j, i) \in \mathcal{E}$}{
Receive $A_i\supsub{\bar{x}}{(k+1)}{i} + \supsub{\bar{\sigma}}{(k+1)}{i}$ and $A_i\supsub{\bar{x}}{(k+1)}{j} + \supsub{\bar{\sigma}}{(k+1)}{j}$\;
Obtain $\supsub{\bar{\lambda}}{(k+1)}{ji}$ by \eqref{eq:optB-updt-lambda}\;
}
K-M updates: $\supsub{\tilde{\psi}}{(k+1)}{} = \supsub{\tilde{\psi}}{(k)}{} + 2\supsub{\gamma}{(k)}{}(\supsub{\bar{\psi}}{(k+1)}{} - \supsub{\psi}{(k+1)}{})$\;
\textbf{Return: }$\{\supsub{x}{(k)}{i}\}$.
\label{algm:multiagt-opt-agg}
\normalsize
\end{algorithm}

\section{Commodity Distribution Problem}

We consider a commodity distribution problem adapted from \cite[Sec.~1.4.3]{facchinei2007finite}, \cite{parise2019distributed}, where several branches of the same company produce a common homogeneous commodity. 
A transport network exists that has markets as its nodes and roads as its edges. 
Denote the node set of this network by $\mathcal{N}_T$ and the edge set $\mathcal{E}_T$, the cardinalities of which are $N_T$ and $E_T$, respectively. 
These branches attempt to cooperatively optimize their total profit by deciding the production quantities at the factories and the distribution quantities over the markets. 

Each branch $i$ delivers the commodity from the factories, denoted by $\mathcal{N}_{T_i}$, to different markets through the transport network. 
Let $N_{T_i} \coloneqq | \mathcal{N}_{T_i} |$.
Its decision vector $x_i \in \rset{n_i}{}$ with $n_i \coloneqq E_T + N_{T_i}$ consists of two parts: $u_i \in \rset{E_T}{+}$ represents the quantities of commodity transported through each road $e_T \in \mathcal{E}_T$; $\nu_i \in \rset{N_{T_i}}{+}$ denotes the quantities of commodity produced by the factories owned by branch $i$. 
These two parts uniquely determine the distribution of commodity over the markets. 
Assuming the factories owned by branch $i$ have maximum production capacities $b_i \in \rset{N_{T_i}}{++}$, each entry of the vector $u_i \in \rset{E_T}{}$ is upper-bounded by $\norm{b_i}_1$.
Denote by $B_T \in \rset{N_T \times E_T}{}$ the incidence matrix of this transport network, and by $E_i \in \rset{N_T \times N_{T_i}}{}$ the indicator matrix which maps from each entry of $\nu_i$ to the corresponding markets. 
Then, we have the local matrix $A_i \coloneqq [B_T, E_i]$ and the local feasible set 
$\mathcal{X}_i \coloneqq \{x_i \in \rset{E_T + N_{T_i}}{} | \bzero \leq \nu_i \leq b_i, 
\bzero \leq u_i \leq \norm{b_i}_1 \otimes \bone_{E_T}, A_ix_i \geq \bzero\}$. 

The objective function of branch $i \in \agentN$ is given by
$J_i(x_i, x_{-i}) = \frac{1}{2}{x_i}^TQ_ix_i + \alpha_i\norm{Ax}^2_2 - (w - \Sigma Ax)^TA_ix_i$, 
where $A \coloneqq [A_1, \ldots, A_N]$, $x \coloneqq [x_1; \cdots; x_N]$, $Q_i \in \pdset{n_i}{++}$ is a diagonal matrix, and we let $w \in \rset{N_T}{++}$ denote the initial unit price, and $\Sigma \in \pdset{N_T}{++}$ the decreasing rate of unit price. 
We further assume that there is a maximum capacity $c \in \rset{N_T}{++}$ for the commodity sold at different markets, and the global constraints are accordingly defined as 
${\textstyle\sum}_{i \in \agentN} A_ix_i \leq c$. 

\begin{figure}
    \centering
    \includegraphics[width=0.42\textwidth]{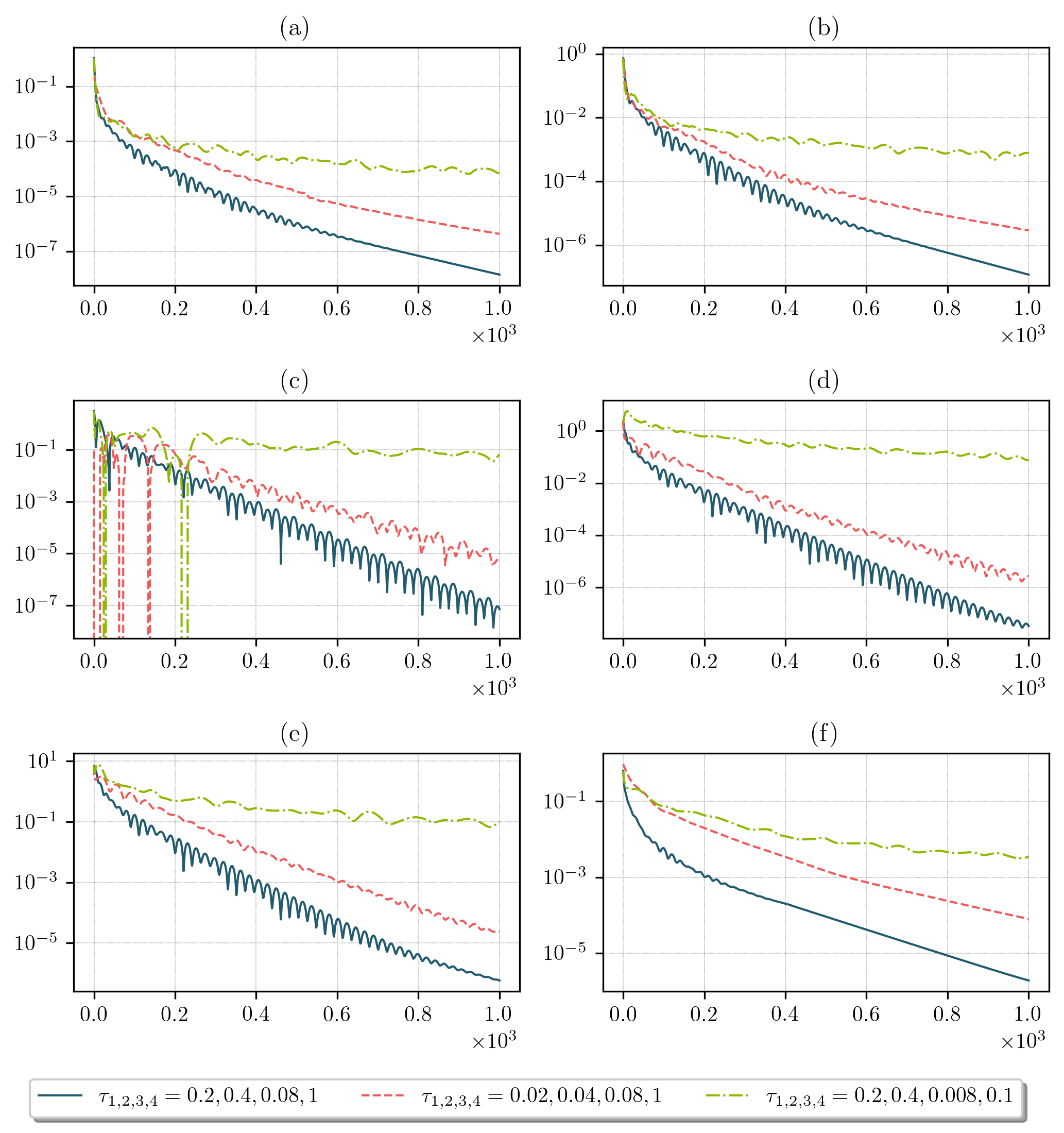}
    \caption{\small{Performances of Algorithm \ref{algm:multiagt-opt-agg} (a) $\sum_{\scriptscriptstyle i \in \agentN} \norm{\supsub{x}{(k+1)}{i} - \supsub{x}{(k)}{i}} / N\supsub{x}{(k)}{i}$, (b) $\sum_{\scriptscriptstyle i \in \agentN} \norm{\supsub{y}{(k+1)}{i} - \supsub{y}{(k)}{i}} / N\norm{\supsub{y}{(k)}{i}}$, (c) $\norm{\max\{0, \sum_{\scriptscriptstyle i\in\agentN}\supsub{A}{}{i}\supsub{x}{(k)}{i} - c\}}$, (d) $\sum_{\scriptscriptstyle (j, i) \in \edgeE} \norm{\supsub{y}{(k+1)}{ji} - \supsub{y}{(k)}{j}}/E$, (e) $\sum_{\scriptscriptstyle i \in \agentN} \norm{\supsub{\sigma}{(k)}{i} - \sum_{\scriptscriptstyle j \in \neighbN{}{-i}}\supsub{A}{}{j}\supsub{x}{(k)}{j}}/N$, (f)$\;\;\sum_{\scriptscriptstyle i \in \agentN} \norm{\supsub{x}{(k+1)}{i} - \supsub{x}{*}{i}} / N\norm{\supsub{x}{*}{i}}$}.}
    \label{fig:perf-multimetrics}
\end{figure}

\begin{figure}
    \centering
    \includegraphics[width=0.4\textwidth]{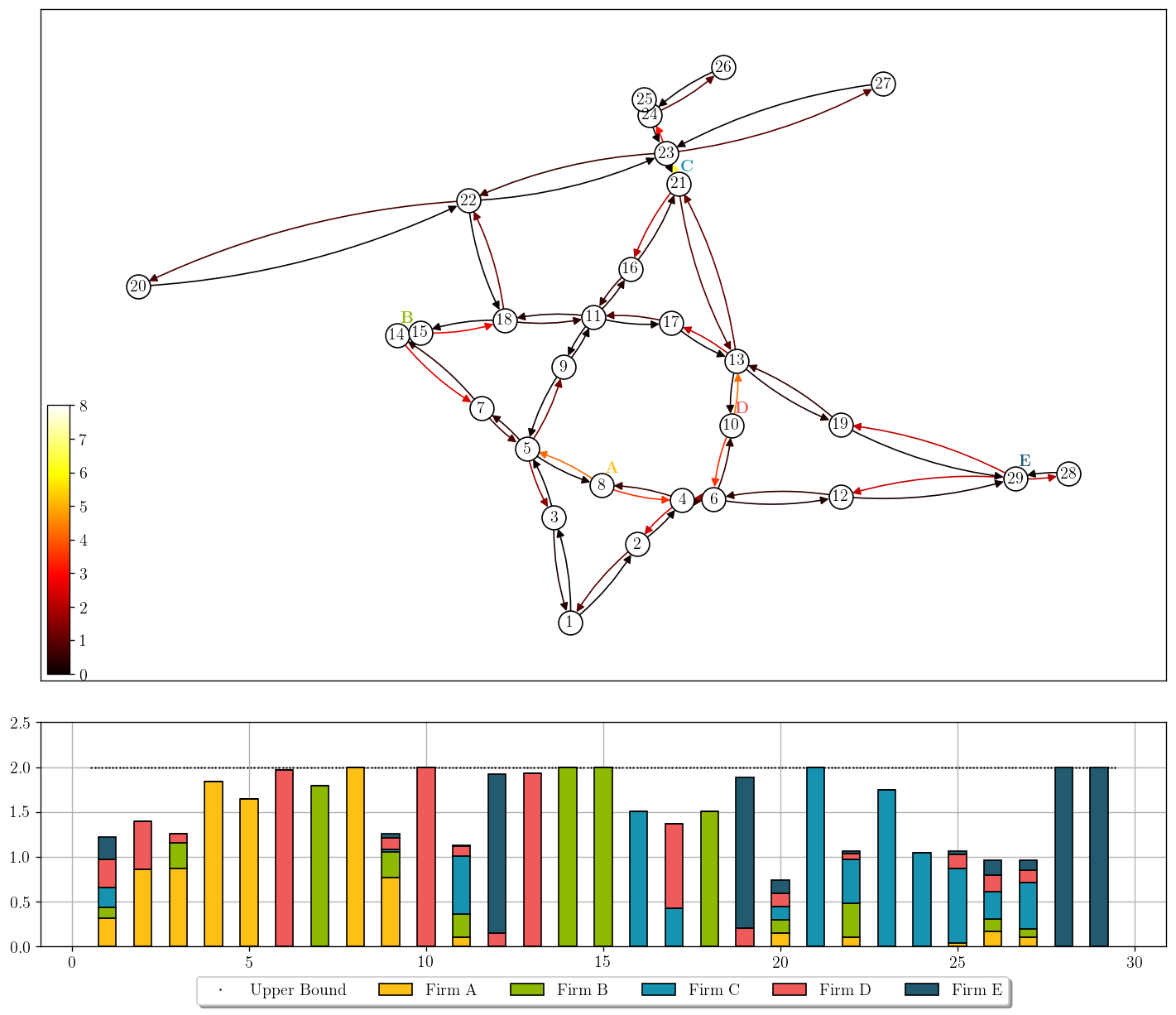}
    \caption{The Commodity Distribution over the Transport Net}
    \label{fig:traffic-plot}
\end{figure}

We use the transport network of the city of Oldenburg  \cite{brinkhoff2002framework}: it consists of $N_T = 29$ nodes (markets) and $E_T = 2 \times 34$ edges (roads). 
Five branches ($N=5$) participate in this problem, each owning a single factory at the given location ($\{8, 14, 21, 10, 29\}$). 
Each factory has a maximum production capacity uniformly randomly chosen from $[10, 14]$. 
The normalized length $\eta_{e_T}$ is defined as the ratio between the length of road $e_T$ and the maximum road length. 
Each diagonal entry of matrix $Q_i$ corresponding to $u_i$ is set as $7\eta_{e_T}$, 
while the entries corresponding to $\nu_i$ are fixed to be $2.8$. 
Moreover, $w = 36 \cdot \bone_{N_T}$, and $\Sigma$ is a matrix with $[\Sigma]_{ii} = 0.23$, for all $i \in \mathcal{N}_T$, $[\Sigma]_{ji} = 0.069 \cdot (1-\eta_{(j, i)})$, for all $(j, i) \in \mathcal{E}_T$, and otherwise zero. 
We also let $\alpha_i \coloneqq 0.21$ and $c \coloneqq 2 \cdot \bone_{N_T}$. 
We can verify numerically that each $J_i$ is jointly convex in $x_i$ and $s_i$ as required in Assumption~\ref{asp:subgrad}. 
The communication network $\commg$ contains a directed circle and $\floor*{\frac{N}{2}} = 2$ randomly selected directed edges. 

The package CVXPY 1.0.31 \cite{diamond2016cvxpy, agrawal2018rewriting} is used to solve the constrained optimization problem for each agent. 
We illustrate the performances of Algorithm \ref{algm:multiagt-opt-agg} in Fig.\ref{fig:perf-multimetrics}. Fig.\ref{fig:perf-multimetrics}(a)(b) describe the relative updating step sizes at each iteration for $x_i$ and $y_i$. 
Fig.\ref{fig:perf-multimetrics}(c)(d)(e) illustrate how the feasibility conditions are satisfied at each iteration.  Fig.\ref{fig:perf-multimetrics}(e) shows the normalized distance between the decision variables obtained by the proposed algorithm and the unique minimizer. 
Fig.\ref{fig:traffic-plot} visualizes the solution computed by Algorithm \ref{algm:multiagt-opt-agg}, which is reflected by the commodity distribution over the markets and the quantity of commodity transported through each road. 
These numerical results verify the validity of Algorithm~\ref{algm:multiagt-opt-agg}, and show a linear convergence rate towards a solution. 

\section{Conclusions and Future Directions}

We propose a distributed solution for multi-agent optimization problems that are globally coupled by aggregates. 
Although we only discuss the cases with homogeneous aggregates in objectives and constraints, with the introduction of a new set of auxiliary variables and its associated local estimates, the proposed algorithm can solve problems with heterogeneous aggregates.
Moreover, the results can be directly extended to handle the problems with convex constraints. 
One of our future directions is to extend the primal decomposition method to the non-cooperative setting. 
Even though the analysis in this paper can be applied to non-cooperative aggregative games with minor modifications, yet to ensure the convergence of the algorithm candidate, we need to postulate that the invoked extended pseudogradient operators are maximally monotone, which holds if and only if the partial derivatives of local objectives w.r.t. local decisions do not depend on others' decisions \cite{belgioioso2018douglas, ye2017distributed, shi2017lana}. 
This condition dramatically restricts the applicability of the proposed method. 
In addition, we note that even under the cooperative setting, the solvable multi-agent optimization problems should have their local objectives jointly convex in the local decisions and the aggregates. 
Another future direction would lie in relaxing this assumption and enable the proposed algorithm to handle a wider range of problems. 

\section*{Acknowledgements}
The authors thank the reviewers for their suggestions and comments, in particular for pointing out the applicability issue of one of our previous assumptions. 
To address this comment, we modify the problem setting from non-cooperative games to cooperative cases. 
Extension to the non-cooperative setting is one of our future directions. 

\appendices
\section*{Appendix}
\renewcommand{\thesubsection}{\Alph{subsection}}

\subsection{Proof of Lemma \ref{le:equiv-initial}}
\label{pf:equiv-initial}
\begin{proof}
Under Assumptions~\ref{asp:subgrad} and \ref{asp:fesbset}, we can equivalently recast the optimization problems in \eqref{eq:probsetup3-1} as a set of inclusions based on the KKT conditions given by: 
\begin{align}\label{eq:dist-kkt-pf}
\begin{split}
& \partial_{x_i} J_i(x_i, \sigma_i) + A_i^T\lambda_i + \sum_{j \in \neighbN{}{-i}}(-A_i^T d_j) + N_{\mathcal{X}_i}(x_i) \ni \bzero \\
& \partial_{\sigma_i} J_i(x_i, \sigma_i) + \lambda_i + d_i \ni \bzero \\
& \bzero \leq \lambda_i \perp c - (A_ix_i + \sigma_i) \geq \bzero,
\end{split}
\end{align}
for each $i \in \agentN$. 
Here, $\lambda_i$ is the Lagrange multiplier enforcing the resource constraints $A_ix_i + \sigma_i \leq c$; $d_i$ is the multiplier enforcing the correct aggregate estimation, i.e., $\sigma_i = \sum_{j \in \neighbN{}{-i}}A_jx_j$. 
Notably, besides the explicit local problem formulation of agent $i$ in \eqref{eq:probsetup3-1}, the decision vector $x_i$ is also involved in the constraints $\sigma_j = \sum_{p \in \neighbN{}{-j}}A_px_p$ for all $j \in \neighbN{}{-i}$, and that is why we need to incorporate $\sum_{j \in \neighbN{}{-i}}(-A_i^T d_j)$ into the first inclusion of \eqref{eq:dist-kkt-pf}. 

We assume that for each agent $i \in \agentN$, its local optimization problem \eqref{eq:probsetup3-1} admits a minimizer $x^*_i$. 
Let $x^* \coloneqq [x^*_i]_{i \in \agentN}$.
Then for each $i \in \agentN$, there exists some $\lambda^*_i \in \rset{l}{}$ and $d^* \coloneqq [d^*_i]_{i \in \agentN} \in \rset{Nl}{}$, such that $(x^*_i, \sigma^*_i, \lambda^*_i, d^*)$ is a solution of the KKT system in \eqref{eq:dist-kkt-pf}, with $\sigma^*_i = \sum_{j \in \neighbN{}{-i}}A_jx^*_j$. 
Since the estimation constraints of $\{\sigma^*_i\}_{i \in \agentN}$ are commonly shared by all agents, the vector $d^*$, which is the stack of all local Lagrange multipliers $\{d_i\}_{i \in \agentN}$ for the estimation constraints, keeps identical among all agents. 
Based on the second inclusion of \eqref{eq:dist-kkt-pf}, for each $j \in \neighbN{}{-i}$, the following inclusion holds: 
\begin{align}
    -A_i^Td^*_j \in A_i^T\partial_{\sigma_j} J_j(x^*_j, \sigma^*_j) + A^T_i\lambda^*_j.
\end{align}
Substituting each $-A^T_id^*_j$ in the first inclusion of \eqref{eq:dist-kkt-pf} with the R.H.S.  of the above inclusion yields: 
\begin{align*}
 \partial_{x_i} J_i(x^*_i, \sigma^*_i) + &{\textstyle\sum}_{j \in \neighbN{}{-i}}A^T_i\partial_{\sigma_i}J_j(x^*_j, \sigma^*_j) \\
& + A^T_i({\textstyle\sum}_{j \in \agentN} \lambda^*_j) + N_{\mathcal{X}_i}(x^*_i) \ni \bzero,
\end{align*}
for each $i \in \agentN$. 
Moreover, since $\bzero \leq \lambda^*_i \perp c - (A_ix^*_i + \sigma^*_i) \geq \bzero$ and $\sigma^*_i = \sum_{j \in \neighbN{}{-i}}A_jx^*_j$ hold for each $i \in \agentN$, the following results trivially follows:
\begin{align*}
\bzero \leq \sum_{i \in \agentN} \lambda^*_i \perp c - \sum_{i \in \agentN} A_ix^*_i \geq \bzero,
\end{align*}
which completes the proof that $(x^*, \sum_{i \in \agentN} \lambda^*_i)$ is a solution of \eqref{eq:kkt-mini}. 
Conversely, if the KKT system in \eqref{eq:kkt-mini} admits a solution $(x^\dagger, \lambda^\dagger)$, where $x^\dagger \coloneqq [x^\dagger_i]_{i \in \agentN}$ and $\lambda^\dagger \in \rset{l}{}$. 
For each agent $i \in \agentN$, we let $\sigma^\dagger_i \coloneqq \sum_{j \in \neighbN{}{-i}} A_jx^\dagger_j$, $\lambda^\dagger_i \coloneqq \frac{1}{N}\lambda^\dagger$, and $d^\dagger_i \in -\partial_{\sigma_i}J_i(x^\dagger_i, \sigma^\dagger_i) - \lambda^\dagger_i$. 
It is obvious that $(x^\dagger_i, \sigma^\dagger_i, \lambda^\dagger_i, d^\dagger)$ with $d^\dagger \coloneqq [d^\dagger_i]_{i \in \agentN}$ is a solution of \eqref{eq:dist-kkt-pf}. 
Therefore, $(x^\dagger_i, \sum_{j \in \neighbN{}{-i}}A_jx^\dagger_j)$ is a minimizer of the local optimization problem~\eqref{eq:probsetup3-1} for each $i \in \agentN$. 
\end{proof}

\subsection{Proof of Theorem \ref{thm:agggamecnst}}
\label{pf:agggamecnst}

\begin{proof}
Suppose $\zer{\optT} \neq \varnothing$ and consider the vector $\psi^* = [x^*; \sigma^*; \lambda^*; \omega^*] \in \zer{\optT}$. 
Firstly, the rows in $(M_y' \cdot \omega^*)$ corresponding to $\mu$ being zero implies the consensus between $y_{ji}^*$ and $y_j^*$ for all $(j, i) \in \edgeE$. 
Furthermore, based on the indicator function $\iota_{\feaset_i}(x_i^*, \sigma_i^*, \by_i^*)$ and the third row of the operator $\optT$ to be zero, i.e., $-B_l^TAx^*-B_l^T\sigma^* = 0$, the following two equations must hold for each agent $i$:
\begin{equation}\label{eq:sigmaeq1}
\begin{split}
(N-1)A_i x_i^* - \sigma_i^* &= \axlW_{ii} y_i^* + {\textstyle\sum}_{j \in \neighbN{+}{i}}\axlW_{ji} y_{j}^* \\ 
A_i x_i^* + \sigma_i^* &= A_j x_j^* + \sigma_j^*, \forall j \in \agentN_{i}^{+}.
\end{split}
\end{equation}
Summing over the L.H.S. of the first equation in (\ref{eq:sigmaeq1}) across all agents, we have
\small
\begin{equation}
\begin{split}
&\sum_{j \in \agentN} (N-1)A_j x_j^* - \sigma_j^*  \\
& = \sum_{j \in \agentN} (N-1)A_j x_j^* + \sum_{j \in \agentN} \Big(A_j x_j^* - A_i x_i^* - \sigma_i^*\Big)
\\ 
&= N \sum_{j \in \agentN} A_j x_j^* - NA_i x_i^* - N\sigma_i^* 
= N(\sum_{j \in \agentN_{-i}}A_j x_j^* - \sigma_i^*),
\end{split}
\end{equation}
\normalsize
where the first equality comes from the second equation in (\ref{eq:sigmaeq1}).
Similarly, adding up the R.H.S. of the first equation in (\ref{eq:sigmaeq1}) yields,
\small
\begin{equation}
\begin{split}
\sum_{i \in \agentN} \Big(\axlW_{ii}y_i^* + \sum_{j \in \neighbN{+}{i}}\axlW_{ji}y_{j}^*\Big)
&= (\bone_N^T \otimes I_l) (\axlW^T \otimes I_l)y^* \\
&= ((\axlW\cdot\bone_N)^T \otimes I_l)y^* = \bzero_l,
\end{split}
\end{equation}
\normalsize
where $y^* = [y^*_1; \cdots; y^*_N]$. As a result, for all $i \in \agentN$, $\sigma_i^* = \sum_{j \in \neighbN{}{-i}}A_jx_j^*$.
With these results, if we check the rows in $\optT$ w.r.t. $x_i$, $\sigma_i$, and $\by_i$, we can obtain the following inclusions: 
\begin{equation}\label{eq:zeroincl3}
\begin{split}
\begin{bmatrix} \partial_{x_i}J_i(x_i^*, \sigma_i^*) + N_{\mathcal{X}_i}(x_i^*) \\ \partial_{\sigma_i}J_i(x^*_i, \sigma^*_i) \\ \bzero \\ \bzero \end{bmatrix}
+ \begin{bmatrix} (N-1)A_i^T \\ -I \\ -\Tilde{W}_{ii} \otimes I_l \\ -[\Tilde{W}_{ji}]_{j \in \neighbN{+}{i}} \otimes I_l  \end{bmatrix}&d_{1i} \\
+ \begin{bmatrix} A_i^T \\ I \\ \bzero \\ \bzero \end{bmatrix} d_{2i}  
+ \begin{bmatrix} A_i^T \\ I \\ \bzero \\ \bzero \end{bmatrix} \lambda^*_{iB}
+ \begin{bmatrix} \bzero \\ \bzero \\ -\sum_{k \in \neighbN{-}{i}} \mu^*_{ik} \\ [\mu^*_{ji}]_{j \in \neighbN{+}{i}} 
 \end{bmatrix} &\ni \bzero,
\end{split}
\end{equation}
where $d_{1i} \in \rset{l}{}$ and $d_{2i} \in \rset{l}{+}$ are the implicitly calculated Lagrange multipliers of the second and first constraints in \eqref{eq:probsetup3-2}, and 
$\lambda^*_{iB} \coloneqq (B_{i\cdot} \otimes I_l)\lambda^*= \sum_{j \in \neighbN{+}{i}}\lambda_{ji}^* - \sum_{j \in \neighbN{-}{i}}\lambda_{ij}^*$ is the explicitly calculated Lagrange multiplier. 
If we check the third and last rows of \eqref{eq:zeroincl3}, we can derive that for all $i \in \agentN$ and for all $k \in \neighbN{-}{i}$,
\begin{equation}
(W_{ii} \otimes I_l)d_i = -{\textstyle\sum}_{k \in \neighbN{-}{i}} \mu^*_{ik}, \;
(W_{ik} \otimes I_l)d_k = \mu^*_{ik}.
\end{equation}
Thus, $(W_{i \cdot} \otimes I_l)\bd = \bzero$, where $\bd = [d_{11}; \cdots; d_{1N}]$. Moreover, the preceding equality holds for all $i \in \agentN$, and we end up with  $(W \otimes I_l)\bd = 0$.
Given that the null space of  $(W \otimes I_l)$ is the consensus subspace, we finally get $d_{11} = \cdots = d_{1N} = d_1$.  
It follows from the second row of the inclusion in (\ref{eq:zeroincl3}) that 
$d_{2j} + \lambda^*_{jB} + \partial_{\sigma_j}J_j(x^*_j, \sigma^*_j) \ni d_{1}$ holds for every $j \in \agentN$. 
By substituting $(N-1)d_1$ in the first row of (\ref{eq:zeroincl3}) with 
$\sum_{j \in \neighbN{}{-i}} d_{2j} + \lambda^*_{jB} + \partial_{\sigma_j}J_j(x^*_j, \sigma^*_j)$ , 
we have the following relation:
\begin{equation}
\begin{split}
& \partial_{x_i}J_i(x_i^*, \sigma_i^*) + A_i^T\sum_{j \in \neighbN{}{-i}} \partial_{\sigma_j}J_j(x^*_j, \sigma^*_j) + N_{\mathcal{X}}(x_i^*) \\ 
&  \qquad + A_i^T(\sum_{j \in \agentN} \lambda^*_{jB}) +  A_i^T(\sum_{j \in \agentN} d_{2j}) \ni \bzero.
\end{split}
\end{equation}
Note that $\sum_{j \in \agentN}\lambda^*_{jB} = \bone^T_{N}(B_l\lambda^*) = \bzero$ by the property of the incident matrix $B_l$, 
and $\partial_{x_i}J(x^*) = \partial_{x_i}J_i(x^*_i, \sigma^*_i) + A_i^T\sum_{j \in \neighbN{}{-i}}\partial_{\sigma_j}J_j(x^*_j, \sigma^*_j)$. 
Let $d_2 \coloneqq \sum_{j \in \agentN} d_{2j}$. 
Consequently, for each agent $i \in \agentN$, we have 
\begin{equation}
\partial_{x_i}J(x^*) + A_i^Td_2 + N_{\mathcal{X}}(x_i^*) \ni \bzero,
\end{equation}
which corresponds to the first KKT condition in \eqref{eq:kkt-mini}. 

Since for each agent $i$, $d_{2i}$ is the Lagrange multiplier of the constraints in $\glbset_i$, it satisfies the following inclusions:
\begin{equation}\label{eq:pridu-cplmtry2}
\bzero \leq d_{2i} \perp A_ix_i^* + \sigma_i^* - c \leq \bzero.
\end{equation}
Notice that the R.H.S. of the above relations keep the same for all agents. 
Adding the L.H.S over all agents yields:
\begin{equation}
\bzero \leq d_2 \perp {\textstyle\sum}_{j \in \agentN} A_jx_j^*- c \leq \bzero, \forall i \in \agentN,
\end{equation}
which corresponds to the second KKT condition in \eqref{eq:kkt-mini}. 
Therefore, for any zeros $\psi^* \in \zer{\optT}$, $x^*$ is the minimizer of the original problem (\ref{eq:probsetup}).

Conversely, suppose there exists at least one solution of the problem (\ref{eq:probsetup3-1}) and denote this solution and its corresponding Lagrange multiplier as $(x^\dagger, \rho^\dagger)$. Then, we can manually set $\sigma_i^\dagger = \sum_{j \in \neighbN{}{-i}}A_jx_j^\dagger$. 
For each agent $i$, we have: 
\begin{equation}\label{eq:prf-mini-kkt}
\begin{split}
& \bzero \in \partial_{x_i}J(x^\dagger) + A_i^T\rho^\dagger + N_{\mathcal{X}_i}(x_i^\dagger) \\
& \bzero \leq \rho^\dagger \perp A_ix_i^\dagger + \sigma_i^\dagger - c \leq \bzero,
\end{split}
\end{equation} 
where $\partial_{x_i}J(x^\dagger)=\partial_{x_i}J_i(x^\dagger_i, \sigma^\dagger_i) + A_i^T\sum_{j \in \neighbN{}{-i}}\partial_{\sigma_j}J_j(x^\dagger_j, \sigma^\dagger_j)$.
We need to prove that there exist $\{y_i^\dagger\}_{i \in \agentN}, y_i^\dagger \in \rset{l}{}$, such that 
\begin{equation}
    (N-1)A_ix_i^\dagger - \sum_{j \in \agentN_{-i}}A_jx_j^\dagger = \axlW_{ii} y_i^\dagger + \sum_{j \in \neighbN{+}{i}} \axlW_{ji} y_j^\dagger, \forall i \in \agentN.
\end{equation}
After reformulating the above equations and concatenating them by row, we obtain:
\begin{equation}\label{eq:auxil-concat-eq}
N[A_ix_i^\dagger]_{i \in \agentN} - \bone_N \otimes ({\textstyle\sum}_{j \in \agentN}A_jx_j^\dagger) = (\axlW^T \otimes I_l) [y^\dagger_i]_{i \in \agentN}.
\end{equation}
Notice that 
\begin{equation}
\begin{split}
& (\bone_N^T \otimes I_l) \cdot (N[A_ix_i^\dagger]_{i \in \agentN} - \bone_N \otimes ({\textstyle\sum}_{j \in \agentN}A_jx_j^\dagger)) \\
& = N{\textstyle\sum}_{j \in \agentN}A_jx_j^\dagger - N \otimes({\textstyle\sum}_{j \in \agentN}A_jx_j^\dagger) = \bzero,
\end{split}
\end{equation}
which implies 
\begin{equation}
    N[A_ix_i^\dagger]_{i \in \agentN} - \bone_N \otimes ({\textstyle\sum}_{j \in \agentN}A_jx_j^\dagger) \in \mathcal{N}(\bone_N^T \otimes I_l).
\end{equation}
Since the communication graph $\commg$ is connected and the null space of $W$ is the range space of $\bone_N$, i.e., $\mathcal{N}(\axlW \otimes I_l) = \mathcal{R}(\bone_N\otimes I_l)$. 
As a result, $\mathcal{R}(\axlW^T \otimes I_l)^\perp = \mathcal{N}(\bone_N^T \otimes I_l)^\perp$, and thus $\mathcal{R}(\axlW^T \otimes I_l) = \mathcal{N}(\bone_N^T \otimes I_l)$. 
Notice that 
$(\axlW^T \otimes I_l)\cdot[y_1; \cdots; y_N]$ spans the range space of $(\axlW^T \otimes I_l)$. There always exists a vector $[y_1^\dagger; \cdots; y_N^\dagger]$, such that the equation (\ref{eq:auxil-concat-eq}) holds. 

After substituting $x_i^*$, $\sigma_i^*$, and $\lambda_i^*$ in (\ref{eq:zeroincl3}) with $x_i^\dagger$, $\sigma_i^\dagger$, and $\lambda_i^\dagger$, it suffices to prove that there exist $d_{1i}, d_{2i}$, and $\lambda^\dagger$ ($i \in \agentN$), such that (i) the inclusions in (\ref{eq:zeroincl3}) hold, (ii)  $d_{2i}$ satisfies $\bzero \leq d_{2i} \perp \sum_{j \in \agentN} A_jx_j^\dagger - c \leq \bzero$. 
We can let $d_{2i} = \frac{1}{N}\rho^\dagger$, $d_{1i} = d_1$ and for all $i \in \agentN$ and choose proper $\lambda^\dagger$, such that 
\begin{equation}
[\partial_{\sigma_i}J_i(x_i^\dagger, \sigma_i^\dagger)]_{i \in \agentN} + \bone_N \otimes (-d_1 + \frac{1}{N}\rho^\dagger) + B_l\lambda^\dagger \ni \bzero.
\end{equation}
The existences of $d_1$ and $\lambda^\dagger$ can be guaranteed by the fact that $(\bone_N \otimes (-d_1) + B_l\lambda^\dagger)$ spans the whole $\rset{Nl}{}$. 
Since $(x^\dagger, \rho^\dagger)$ satisfies the first condition in \eqref{eq:prf-mini-kkt}, it can be verified that with the chosen $d_1$, $\frac{1}{N}\rho^\dagger$ and $\lambda^\dagger$, the following relation holds for each agent $i$:
\begin{equation}
\partial_{x_i}J_i(\supsub{x}{\dagger}{i}, \supsub{\sigma}{\dagger}{i}) + A_i^T((N-1)d_1 + \frac{1}{N}\rho^\dagger+\supsub{\lambda}{\dagger}{iB}) + N_{\mathcal{X}_i}(\supsub{x}{\dagger}{i}) \ni \bzero.
\end{equation}
Accordingly, $\mu^\dagger$ can be obtained by the forth row of (\ref{eq:zeroincl3}). As a result, $[x^\dagger; \sigma^\dagger; \lambda^\dagger; \omega^\dagger] \in \zer{\optT}$.
\end{proof}

\bibliographystyle{IEEEtran}
\bibliography{IEEEabrv,references}

\begin{thebibliography}{10}
\providecommand{\url}[1]{#1}
\csname url@samestyle\endcsname
\providecommand{\newblock}{\relax}
\providecommand{\bibinfo}[2]{#2}
\providecommand{\BIBentrySTDinterwordspacing}{\spaceskip=0pt\relax}
\providecommand{\BIBentryALTinterwordstretchfactor}{4}
\providecommand{\BIBentryALTinterwordspacing}{\spaceskip=\fontdimen2\font plus
\BIBentryALTinterwordstretchfactor\fontdimen3\font minus
  \fontdimen4\font\relax}
\providecommand{\BIBforeignlanguage}[2]{{%
\expandafter\ifx\csname l@#1\endcsname\relax
\typeout{** WARNING: IEEEtran.bst: No hyphenation pattern has been}%
\typeout{** loaded for the language `#1'. Using the pattern for}%
\typeout{** the default language instead.}%
\else
\language=\csname l@#1\endcsname
\fi
#2}}
\providecommand{\BIBdecl}{\relax}
\BIBdecl

\bibitem{yang2019survey}
T.~Yang, X.~Yi, J.~Wu, Y.~Yuan, D.~Wu, Z.~Meng, Y.~Hong, H.~Wang, Z.~Lin, and
  K.~H. Johansson, ``A survey of distributed optimization,'' \emph{Annual
  Reviews in Control}, vol.~47, pp. 278--305, 2019.

\bibitem{nedic2018distributed}
A.~Nedi{\'c} and J.~Liu, ``Distributed optimization for control,'' \emph{Annual
  Review of Control, Robotics, and Autonomous Systems}, vol.~1, pp. 77--103,
  2018.

\bibitem{tang2020zeroth}
Y.~Tang, Z.~Ren, and N.~Li, ``Zeroth-order feedback optimization for
  cooperative multi-agent systems,'' in \emph{2020 59th IEEE Conference on
  Decision and Control (CDC)}.\hskip 1em plus 0.5em minus 0.4em\relax IEEE,
  2020, pp. 3649--3656.

\bibitem{hu2018distributed}
J.~Hu, Y.~Xiao, and J.~Liu, ``Distributed algorithms for solving locally
  coupled optimization problems on agent networks,'' in \emph{2018 IEEE
  Conference on Decision and Control (CDC)}.\hskip 1em plus 0.5em minus
  0.4em\relax IEEE, 2018, pp. 2420--2425.

\bibitem{jensen2010aggregative}
M.~K. Jensen, ``Aggregative games and best-reply potentials,'' \emph{Economic
  theory}, vol.~43, no.~1, pp. 45--66, 2010.

\bibitem{barrera2014dynamic}
J.~Barrera and A.~Garcia, ``Dynamic incentives for congestion control,''
  \emph{IEEE Transactions on Automatic Control}, vol.~60, no.~2, pp. 299--310,
  2014.

\bibitem{wei2021mechanism}
X.~Wei and A.~Anastasopoulos, ``Mechanism design for demand management in
  energy communities,'' \emph{Games}, vol.~12, no.~3, p.~61, 2021.

\bibitem{ma2011decentralized}
Z.~Ma, D.~S. Callaway, and I.~A. Hiskens, ``Decentralized charging control of
  large populations of plug-in electric vehicles,'' \emph{IEEE Transactions on
  control systems technology}, vol.~21, no.~1, pp. 67--78, 2011.

\bibitem{camisa2021distributed}
A.~Camisa, F.~Farina, I.~Notarnicola, and G.~Notarstefano, ``Distributed
  constraint-coupled optimization via primal decomposition over random
  time-varying graphs,'' \emph{Automatica}, vol. 131, p. 109739, 2021.

\bibitem{falsone2017dual}
A.~Falsone, K.~Margellos, S.~Garatti, and M.~Prandini, ``Dual decomposition for
  multi-agent distributed optimization with coupling constraints,''
  \emph{Automatica}, vol.~84, pp. 149--158, 2017.

\bibitem{heydaribeni2019distributed}
N.~Heydaribeni and A.~Anastasopoulos, ``Distributed mechanism design for
  network resource allocation problems,'' \emph{IEEE Transactions on Network
  Science and Engineering}, vol.~7, no.~2, pp. 621--636, 2019.

\bibitem{facchinei2010generalized}
F.~Facchinei and C.~Kanzow, ``Generalized {Nash} equilibrium problems,''
  \emph{Annals of Operations Research}, vol. 175, no.~1, pp. 177--211, 2010.

\bibitem{belgioioso2018douglas}
G.~Belgioioso and S.~Grammatico, ``A {Douglas-Rachford} splitting for
  semi-decentralized equilibrium seeking in generalized aggregative games,'' in
  \emph{2018 IEEE Conference on Decision and Control (CDC)}.\hskip 1em plus
  0.5em minus 0.4em\relax IEEE, 2018, pp. 3541--3546.

\bibitem{yi2019operator}
P.~Yi and L.~Pavel, ``An operator splitting approach for distributed
  generalized {Nash} equilibria computation,'' \emph{Automatica}, vol. 102, pp.
  111--121, 2019.

\bibitem{pavel2019distributed}
L.~Pavel, ``Distributed {GNE} seeking under partial-decision information over
  networks via a doubly-augmented operator splitting approach,'' \emph{IEEE
  Transactions on Automatic Control}, vol.~65, no.~4, pp. 1584--1597, 2019.

\bibitem{parise2019distributed}
F.~Parise, B.~Gentile, and J.~Lygeros, ``A distributed algorithm for
  almost-{Nash} equilibria of average aggregative games with coupling
  constraints,'' \emph{IEEE Transactions on Control of Network Systems},
  vol.~7, no.~2, pp. 770--782, 2019.

\bibitem{parise2020distributed}
F.~Parise, S.~Grammatico, B.~Gentile, and J.~Lygeros, ``Distributed convergence
  to {Nash} equilibria in network and average aggregative games,''
  \emph{Automatica}, vol. 117, p. 108959, 2020.

\bibitem{liang2017distributed}
S.~Liang, P.~Yi, and Y.~Hong, ``Distributed {Nash} equilibrium seeking for
  aggregative games with coupled constraints,'' \emph{Automatica}, vol.~85, pp.
  179--185, 2017.

\bibitem{gadjov2020single}
D.~Gadjov and L.~Pavel, ``Single-timescale distributed {GNE} seeking for
  aggregative games over networks via forward-backward operator splitting,''
  \emph{IEEE Transactions on Automatic Control}, 2020.

\bibitem{bianchi2020fast}
M.~Bianchi, G.~Belgioioso, and S.~Grammatico, ``Fast generalized {Nash}
  equilibrium seeking under partial-decision information,'' \emph{arXiv
  preprint arXiv:2003.09335}, 2020.

\bibitem{belgioioso2020distributed}
G.~Belgioioso, A.~Nedi{\'c}, and S.~Grammatico, ``Distributed generalized
  {Nash} equilibrium seeking in aggregative games on time-varying networks,''
  \emph{IEEE Transactions on Automatic Control}, vol.~66, no.~5, pp.
  2061--2075, 2020.

\bibitem{facchinei2007finite}
F.~Facchinei and J.-S. Pang, \emph{Finite-dimensional variational inequalities
  and complementarity problems}.\hskip 1em plus 0.5em minus 0.4em\relax
  Springer Science \& Business Media, 2007.

\bibitem{BauschkeHeinzH2017CAaM}
H.~H. Bauschke, \emph{\BIBforeignlanguage{eng}{{Convex Analysis and Monotone
  Operator Theory in Hilbert Spaces}}}, 2nd~ed., ser. CMS Books in Mathematics,
  Ouvrages de mathématiques de la SMC, 2017.

\bibitem{bell1965gershgorin}
H.~E. Bell, ``Gershgorin's theorem and the zeros of polynomials,'' \emph{The
  American Mathematical Monthly}, vol.~72, no.~3, pp. 292--295, 1965.

\bibitem{brinkhoff2002framework}
T.~Brinkhoff, ``A framework for generating network-based moving objects,''
  \emph{GeoInformatica}, vol.~6, no.~2, pp. 153--180, 2002.

\bibitem{diamond2016cvxpy}
S.~Diamond and S.~Boyd, ``{CVXPY}: {A} {P}ython-embedded modeling language for
  convex optimization,'' \emph{Journal of Machine Learning Research}, vol.~17,
  no.~83, pp. 1--5, 2016.

\bibitem{agrawal2018rewriting}
A.~Agrawal, R.~Verschueren, S.~Diamond, and S.~Boyd, ``A rewriting system for
  convex optimization problems,'' \emph{Journal of Control and Decision},
  vol.~5, no.~1, pp. 42--60, 2018.

\bibitem{ye2017distributed}
M.~Ye and G.~Hu, ``Distributed nash equilibrium seeking by a consensus based
  approach,'' \emph{IEEE Transactions on Automatic Control}, vol.~62, no.~9,
  pp. 4811--4818, 2017.

\bibitem{shi2017lana}
W.~Shi and L.~Pavel, ``Lana: an admm-like nash equilibrium seeking algorithm in
  decentralized environment,'' in \emph{2017 American Control Conference
  (ACC)}.\hskip 1em plus 0.5em minus 0.4em\relax IEEE, 2017, pp. 285--290.

\end{thebibliography}

\end{document}